\documentclass[11pt]{article}
\usepackage[left=1in,top=1in,right=1in,bottom=1in]{geometry}
\usepackage{times}

\usepackage{amsmath}
\usepackage{amssymb}
\usepackage{amsthm}
\usepackage{graphicx}
\usepackage{verbatim}
\usepackage{url}
\usepackage{algorithm}
\usepackage[noend]{algpseudocode}
\allowdisplaybreaks

\usepackage{siunitx}
\usepackage[T1]{fontenc}

\newtheorem{proposition}{Proposition}
\newtheorem{corollary}{Corollary}

\allowdisplaybreaks

\begin{document}

\author{
Sam Ganzfried\\
Ganzfried Research, Cornell University\\
\texttt{sam.ganzfried@gmail.com}
}

\date{\vspace{-5ex}}

\title{Variable Bound Tightening for Nash Equilibrium Computation in Multiplayer Imperfect-Information Games} 

\maketitle

\begin{abstract}
There has been significant recent progress in algorithms for approximation of Nash equilibrium in large two-player zero-sum imperfect-information games and exact computation of Nash equilibrium in multiplayer strategic-form games. While counterfactual regret minimization and fictitious play are scalable to large games and have convergence guarantees in two-player zero-sum games, they do not guarantee convergence to Nash equilibrium in multiplayer games. Recently, an approach has been presented for exact computation of Nash equilibrium in multiplayer imperfect-information games that solves a quadratically constrained program based on a nonlinear complementarity problem formulation derived from the sequence-form game representation. This formulation was solved using Gurobi's nonconvex quadratic solver, which employs spatial branch-and-bound to iteratively refine variable bounds by solving convex relaxations of bilinear terms. During presolve, Gurobi introduces auxiliary variables and, in some cases, binary variables, leading to an internal mixed-integer quadratically constrained program. This approach was demonstrated to outperform prior algorithms from the Gambit software suite and quickly solve three-player Kuhn poker after removal of dominated actions; however, the algorithm was not able to solve the full version of the game within 24 hours. In this paper, we derive finite bounds on slack and multiplier variables in the nonlinear complementarity formulation. These bounds can strengthen the convex relaxations used within spatial branch-and-bound. We compare uniform corollary bounds with finer proposition-based bounds on full three-player Kuhn poker. All eight bounded configurations return numerically verified equilibria in under 17 seconds, while the unbounded formulation finds no feasible solution within 300 seconds. The finest bounds do not consistently yield the fastest computation.
\end{abstract}

\section{Nonlinear complementarity formulation}
\label{se:ncp}
In this section we review the derivation of a nonlinear complementarity program formulation (NCP) for multiplayer Nash equilibrium in imperfect-information games~\cite{Ganzfried26:Quadratic}. Imperfect-information games are modeled using extensive-form game trees, where play proceeds from the root node to a terminal leaf node at which point all players receive payoffs. Each non-terminal node has an associated player (possibly \emph{chance}) that makes the decision at that node. These nodes are partitioned into \emph{information sets}, where the player whose turn it is to move cannot distinguish among the states in the same information set. Therefore, in any given information set, a player must choose actions with the same distribution at each state contained in the information set. If no player forgets information that they previously knew, we say that the game has \emph{perfect recall}. We assume finite games with perfect recall throughout. A \emph{behavioral strategy} for player $i,$ $\sigma_i \in \Sigma_i,$ assigns a probability distribution over actions at each information set belonging to $i$. 

Rather than operate on the full pure strategy space, whose size can be exponential in the size of the game tree, the \emph{sequence-form representation} works with sequences of actions along trajectories from the root node to leaf nodes~\cite{Koller94:Fast}. For player 1, the matrix $\mathbf{E}$ is defined where there is one row per information set plus an initial normalization row, and each column corresponds to an action sequence (including an initial column for the ``empty'' action sequence). In the first row of $\mathbf{E}$ the first element is 1 and all other elements are 0; subsequent rows have -1 for the entries corresponding to the action sequence leading to the root of the information set, and 1 for all actions that can be taken at the information set (and 0 otherwise). Thus $\mathbf{E}$ has dimension $c_1 \times d_1$, where $c_i$ is one plus the number of information sets for player $i$ and $d_i$ is the number of action sequences for player $i$. Matrix $\mathbf{F}$ is defined analogously for player 2. The vector $\mathbf{e}$ is defined to be a column vector of length $c_1$ with 1 in the first position and 0 in other entries, and vector $\mathbf{f}$ is defined with length $c_2$ analogously. The matrix $\mathbf{A}$ is defined with dimension $d_1 \times d_2$ where entry $A_{ij}$ sums player 1's terminal payoffs, weighted by chance reach probabilities, over terminal histories with player sequences $i$ and $j$. The matrix $\mathbf{B}$ of player 2's payoffs is defined analogously. In zero-sum games,
\[
\mathbf{B} = -\mathbf{A}.
\]

Given these matrices we can solve one of two linear programming problems to compute a Nash equilibrium in zero-sum extensive-form games~\cite{Koller94:Fast}. In the first formulation the primal variables $\mathbf{x}$ correspond to player 1's realization plan while the corresponding dual variables encode player 2's realization plan. In the second formulation, which is the dual problem of the first formulation, the primal decision variables $\mathbf{y}$ correspond to player 2's strategy while the dual variables correspond to player 1's strategy.

\[
\begin{array}{rrl} 
&\max_{\mathbf{x},\mathbf{q}}& -\mathbf{q}^T \mathbf{f} \\ 
&\mbox{s.t.}& \mathbf{x}^T (-\mathbf{A}) - \mathbf{q}^T \mathbf{F} \leq \mathbf{0} \\
& & \mathbf{x}^T \mathbf{E}^T = \mathbf{e}^T \\
& & \mathbf{x} \geq \mathbf{0}\\
\end{array} 
\]

\[
\begin{array}{rrl} 
&\min_{\mathbf{y},\mathbf{p}}& \mathbf{e}^T \mathbf{p} \\ 
&\mbox{s.t.}& -\mathbf{A} \mathbf{y} + \mathbf{E}^T \mathbf{p} \geq \mathbf{0} \\
& & -\mathbf{F} \mathbf{y} = -\mathbf{f} \\
& & \mathbf{y} \geq \mathbf{0}\\
\end{array} 
\]

For two-player non-zero-sum games, the problem of finding a Nash equilibrium is the feasibility problem of finding
$\mathbf{x},$ $\mathbf{y},$ $\mathbf{p},$ $\mathbf{q}$, such that~\cite{Koller94:Fast}:

\[
\begin{array}{rrl} 
-\mathbf{A} \mathbf{y} + \mathbf{E}^T \mathbf{p} & \geq & \mathbf{0} \\
-\mathbf{B}^T \mathbf{x} + \mathbf{F}^T \mathbf{q} & \geq &\mathbf{0} \\
-\mathbf{E} \mathbf{x} & = &-\mathbf{e} \\
-\mathbf{F} \mathbf{y} & = &-\mathbf{f} \\
\mathbf{x} & \geq &\mathbf{0}\\
\mathbf{y} &\geq &\mathbf{0}\\
\mathbf{x}^T (-\mathbf{A} \mathbf{y} + \mathbf{E}^T \mathbf{p}) & = &0\\
\mathbf{y}^T (-\mathbf{B}^T \mathbf{x} + \mathbf{F}^T \mathbf{q}) & = &0\\
\end{array} 
\]

The final two constraints are called complementary slackness conditions (CSC), and the full system is known
as a linear complementarity problem (LCP). It is no longer a linear program since the CSCs involve products of variables.
This LCP can be solved using Lemke's algorithm~\cite{Lemke65:Bimatrix}; the related Lemke-Howson algorithm applies to bimatrix games~\cite{Lemke64:Equilibrium}.
We would like to extend this result to develop a feasibility problem for games with $n > 2$ players. We will start with the case $n = 3.$ The sequence-form representation extends straightforwardly to 3 players. We define matrix $\mathbf{G}$ for player 3 analogously to $\mathbf{E}$ and $\mathbf{F},$ and define vector $\mathbf{g}$ analogously to $\mathbf{e},\mathbf{f}.$ The utility functions can no longer be represented as 2-dimensional matrices. We write $u_1(i,j,k)$ for the sum of player 1's terminal payoffs weighted by chance reach probabilities over terminal histories with player sequences $i,j,k$. These tensor entries are distinct from the terminal payoffs before chance weighting. We represent player 2 and 3's utilities analogously as $u_2(i,j,k), u_3(i,j,k).$

Consider the problem of player 1 playing a best response when player 2 plays $\mathbf{y}$ and player 3 plays $\mathbf{z}$:

\[
\begin{array}{rrl} 
&\max_{\mathbf{x}} &\sum_i \sum_j \sum_k x_i y_j z_k u_1(i,j,k)  \\ 
&\mbox{s.t.}& \mathbf{E} \mathbf{x} = \mathbf{e} \\
& & \mathbf{x} \geq \mathbf{0}\\
\end{array} 
\]

Let us rewrite this as a convex minimization problem by negating the objective:

\[
\begin{array}{rrl} 
&\min_{\mathbf{x}} &-\sum_i \sum_j \sum_k x_i y_j z_k u_1(i,j,k)  \\ 
&\mbox{s.t.}& \mathbf{E} \mathbf{x} = \mathbf{e} \\
& & \mathbf{x} \geq \mathbf{0}\\
\end{array} 
\]

The Lagrangian is
\[
\begin{aligned}
L(\mathbf{x},\boldsymbol{\tau}^1,\mathbf{r}^1)
&= -\sum_i \sum_j \sum_k x_i y_j z_k u_1(i,j,k) \\
&\quad - (\mathbf{E} \mathbf{x} - \mathbf{e})^T \boldsymbol{\tau}^1
- (\mathbf{r}^1)^T \mathbf{x}.
\end{aligned}
\]

Defining $\boldsymbol{\lambda}^1 = -\boldsymbol{\tau}^1$:

\[
\begin{aligned}
L(\mathbf{x},\boldsymbol{\lambda}^1,\mathbf{r}^1)
&= -\sum_i \sum_j \sum_k x_i y_j z_k u_1(i,j,k) \\
&\quad - (\mathbf{E} \mathbf{x} - \mathbf{e})^T (-\boldsymbol{\lambda}^1)
- (\mathbf{r}^1)^T \mathbf{x} \\
&= -\sum_i \sum_j \sum_k x_i y_j z_k u_1(i,j,k) \\
&\quad + (\mathbf{E} \mathbf{x} - \mathbf{e})^T \boldsymbol{\lambda}^1
- (\mathbf{r}^1)^T \mathbf{x}.
\end{aligned}
\]

\[
\begin{aligned}
\frac{\partial L}{\partial x_i}
&= -\sum_j \sum_k y_j z_k u_1(i,j,k) \\
&\quad + \sum_{j} \lambda^1_j E_{ji} - r^1_i.
\end{aligned}
\]
 
So the first-order necessary optimality conditions for player 1's best response problem are:

\[
\begin{gathered}
\mathbf{E}\mathbf{x} = \mathbf{e},\qquad
\mathbf{x} \ge \mathbf{0},\qquad
\mathbf{r}^1 \ge \mathbf{0}, \\
\begin{aligned}
&-\sum_j \sum_k y_j z_k u_1(i,j,k) \\
&\qquad + \sum_j \lambda^1_j E_{ji} - r^1_i = 0
\quad \text{for all } i,
\end{aligned} \\
x_i r^1_i = 0 \quad \text{for all } i.
\end{gathered} 
\]

Since the problem is a convex minimization problem these conditions are sufficient for optimality as well.
Adding in analogous constraints for players 2 and 3, the problem is to find $\mathbf{x},\mathbf{y},\mathbf{z}$,
$\boldsymbol{\lambda}^1,\boldsymbol{\lambda}^2,\boldsymbol{\lambda}^3$, $\mathbf{r}^1,\mathbf{r}^2,\mathbf{r}^3$ such that:

\begin{align}
\label{eq:ncp}
&-\sum_j \sum_k y_j z_k u_1(i,j,k) \nonumber \\
&\qquad + \sum_{j} \lambda^1_j E_{ji} - r^1_i = 0 \quad \text{for all } i \nonumber \\
&-\sum_i \sum_k x_i z_k u_2(i,j,k) \nonumber \\
&\qquad + \sum_{i} \lambda^2_i F_{ij} - r^2_j = 0 \quad \text{for all } j \nonumber \\
&-\sum_i \sum_j x_i y_j u_3(i,j,k) \nonumber \\
&\qquad + \sum_{j} \lambda^3_j G_{jk} - r^3_k = 0 \quad \text{for all } k \nonumber \\
&\mathbf{E} \mathbf{x} = \mathbf{e} \nonumber \\
&\mathbf{F} \mathbf{y} = \mathbf{f} \nonumber \\
&\mathbf{G} \mathbf{z} = \mathbf{g} \nonumber \\
&\mathbf{x} \geq \mathbf{0} \nonumber \\
&\mathbf{y} \geq \mathbf{0}  \\
&\mathbf{z} \geq \mathbf{0} \nonumber \\
&\mathbf{r}^1 \geq \mathbf{0} \nonumber \\
&\mathbf{r}^2 \geq \mathbf{0} \nonumber \\
&\mathbf{r}^3 \geq \mathbf{0} \nonumber \\
&x_i r^1_i = 0 \quad \text{for all } i \nonumber \\
&y_i r^2_i = 0 \quad \text{for all } i \nonumber \\
&z_i r^3_i = 0 \quad \text{for all } i \nonumber
\end{align}

This feasibility program is not a linear program due to several products of variables:
$y_j z_k,$ $x_i z_k$, $x_i y_j$, $x_i r^1_i,$ $y_i r^2_i,$ $z_i r^3_i.$ 
This formulation can be straightforwardly generalized to $n > 3$ players by adding in the original decision constraints,
slack nonnegativity conditions, Lagrange derivative conditions, and complementary slackness conditions for each player.
The Lagrange derivative conditions will involve products of $n-1$ variables, while the complementary slackness conditions 
still involve products of 2 variables. However, we can still model the Lagrange derivative conditions with quadratic constraints by incrementally defining new product variables. E.g., define $w_{ij} = x_i y_j,$ then $v_{ijk} = w_{ij} z_k$, etc. So for arbitrary $n > 2$, we can define the problem of finding a Nash equilibrium as a nonlinear complementarity problem, which can be modeled as a quadratically-constrained feasibility program. For efficiency we prefer to introduce as few new auxiliary bilinear variables as possible. For example, for four players with variables $\mathbf{x}^1,\mathbf{x}^2,\mathbf{x}^3,\mathbf{x}^4,$ we can create a new bilinear program with only the introduction of new variables $y_{ij} = x^1_i x^2_j$ and $z_{ij} = x^3_i x^4_j.$ We have $x^1_i x^2_j x^3_k = y_{ij} x^3_k,$ $x^1_i x^2_j x^4_m = y_{ij} x^4_m,$ $x^1_i x^3_k x^4_m = x^1_i z_{km},$ $x^2_j x^3_k x^4_m = x^2_j z_{km}.$ 

\section{Tightening variable bounds}
\label{se:bound}
We can directly solve the NCP formulation (\ref{eq:ncp}) in Gurobi using the nonconvex quadratic solver~\cite{Gurobi26:Gurobi}. Gurobi allows us to input the nonconvex bilinear terms that occur in the first three sets of constraints as well as the final complementarity constraints. The prior formulation~\cite{Ganzfried26:Quadratic} used the following variable domains. For the action sequence variables $\mathbf{x}, \mathbf{y}, \mathbf{z},$ the obvious domain is [0,1]. The lower bounds are directly implied by the constraints in the original problem. The upper bounds of 1 follow from the sequence-form constraints (e.g., $\mathbf{E}\mathbf{x} = \mathbf{e}$), which are essentially flow conservation constraints. For the slackness variables all we know is $\mathbf{r}^i \geq \mathbf{0}.$ So the ``natural'' domain to use is $[0,\infty).$ For the unconstrained multiplier variables $\lambda^i_j$ the natural domain is $(-\infty,\infty).$ These were the variable domains used in the prior implementation.

An obvious question to consider is whether these variable domain bounds can be tightened without affecting the solutions, since this could have a dramatic improvement on the performance of spatial branch and bound. Let us start by looking at $r^1_i$ (the other slack variables will be analogous). We can bound $r_i^1$ by exploiting its interpretation as a dual slack variable in player 1's best-response linear program. The argument uses the standard sequence-form constraints, without additional equalities fixing action probabilities.

Fix opponent strategies $y,z$, and define
\[
c_i := \sum_{j,k} y_j z_k u_1(i,j,k).
\]

Consider the player 1 best-response linear program
\[
\max_{\mathbf{x}} \ \mathbf{c}^\top \mathbf{x}
\quad \text{s.t.} \quad
\mathbf{E}\mathbf{x}=\mathbf{e},\;
\mathbf{x}\ge \mathbf{0}.
\]

Let
\[
V_1^{\max}
=
\max\{\mathbf{c}^\top\mathbf{x} :
\mathbf{E}\mathbf{x}=\mathbf{e},
\mathbf{x}\ge\mathbf{0}\},
\]%
\noindent and for each sequence $i$, define
\[
\begin{aligned}
U_i = \max\{\mathbf{c}^\top\mathbf{x} :\;&
\mathbf{E}\mathbf{x}=\mathbf{e}, \\
&\mathbf{x}\ge\mathbf{0},\ x_i=1\}.
\end{aligned}
\]

Let $(\mathbf{x}^*,\boldsymbol{\lambda}^1,\mathbf{r}^1)$ satisfy the corresponding KKT conditions
\[
\begin{gathered}
\mathbf{E}^\top\boldsymbol{\lambda}^1-\mathbf{r}^1=\mathbf{c},
\qquad \mathbf{r}^1\ge \mathbf{0}, \\
\mathbf{E}\mathbf{x}^*=\mathbf{e},
\qquad \mathbf{x}^*\ge \mathbf{0}, \\
x_i^*r_i^1=0\quad\text{for all }i.
\end{gathered}
\]

\begin{proposition}
\label{pr:slack}
For every sequence $i$,
\[
0 \le r_i^1 \le V_1^{\max}-U_i.
\]
\end{proposition}

\begin{proof}
By strong duality, $V_1^{\max}=\mathbf{e}^\top\boldsymbol{\lambda}^1$.

For any feasible realization plan $\mathbf{x}$, we have
\[
\begin{aligned}
\mathbf{c}^\top \mathbf{x}
&= \mathbf{x}^\top(\mathbf{E}^\top\boldsymbol{\lambda}^1-\mathbf{r}^1) \\
&= (\mathbf{E}\mathbf{x})^\top\boldsymbol{\lambda}^1-(\mathbf{r}^1)^\top \mathbf{x} \\
&= \mathbf{e}^\top\boldsymbol{\lambda}^1-(\mathbf{r}^1)^\top \mathbf{x} \\
&= V_1^{\max}-(\mathbf{r}^1)^\top \mathbf{x}.
\end{aligned}
\]

Fix a sequence $i$. Since sequence $i$ is realizable, there exists at least one feasible realization plan satisfying $x_i=1$. Therefore $U_i$ is well-defined.

Now let $\mathbf{x}$ be any feasible realization plan with $x_i=1$.
Since $\mathbf{r}^1\ge \mathbf{0}$ and $\mathbf{x}\ge \mathbf{0}$,
$(\mathbf{r}^1)^\top \mathbf{x} \ge r_i^1 x_i = r_i^1$.

Hence
\[
\mathbf{c}^\top \mathbf{x}
=
V_1^{\max}-(\mathbf{r}^1)^\top \mathbf{x}
\le
V_1^{\max}-r_i^1.
\]

Rearranging gives $r_i^1 \le V_1^{\max}-\mathbf{c}^\top \mathbf{x}$.
Taking the maximum over feasible $\mathbf{x}$ with $x_i=1$ yields
$r_i^1 \le V_1^{\max}-U_i$. The bound $r_i^1\ge 0$ follows from the KKT conditions.
\end{proof}

The quantity $V_1^{\max}-U_i$ represents the loss in best-response value incurred by forcing sequence $i$. Thus, Proposition~\ref{pr:slack} gives a sequence-specific bound, tighter for nearly optimal sequences and looser for highly suboptimal ones.

\begin{corollary}
\label{co:slack}
Let $g_1(h)$ be player 1's payoff at terminal history $h\in H$, before chance weighting. Define
\[
\begin{aligned}
u_1^{\max} &:= \max_{h\in H}g_1(h), \\
u_1^{\min} &:= \min_{h\in H}g_1(h).
\end{aligned}
\]

Then $0\le r_i^1\le u_1^{\max}-u_1^{\min}$.
\end{corollary}

\begin{proof}
Since every expected utility $\mathbf{c}^\top\mathbf{x}$ is a convex combination of terminal payoffs, $u_1^{\min}\le\mathbf{c}^\top\mathbf{x}\le u_1^{\max}$ for every feasible realization plan $\mathbf{x}$. Therefore, $V_1^{\max}\le u_1^{\max}$ and $U_i\ge u_1^{\min}$.

Applying Proposition~\ref{pr:slack} yields
\[
0
\le
r_i^1
\le
V_1^{\max}-U_i
\le
u_1^{\max}-u_1^{\min}. \qedhere
\]
\end{proof}

Although weaker than the sequence-specific bound of Proposition~\ref{pr:slack}, the bound
in Corollary~\ref{co:slack} is trivial to compute and remains significantly tighter than
the original domain $0 \le r_i^1 < \infty.$ This tightening leads to stronger convex relaxations 
of the bilinear constraints, which can significantly improve the efficiency of spatial 
branch-and-bound. Analogous bounds hold for the other slack variables. 

While the multiplier vector $\boldsymbol{\lambda}^1$ need not be unique, valid finite bounds
can still be obtained. Unlike Proposition~\ref{pr:slack}, which yields sequence-specific
bounds that depend on the opponent realization plans, we seek bounds that are
valid globally and therefore independent of the unknown equilibrium strategies.

Define
\[
M_1 := \max_{h\in H}|g_1(h)|,
\qquad R_1 := u_1^{\max}-u_1^{\min}.
\]

Write $s_\ell(h)$ for player $\ell$'s terminal sequence and $\pi_0(h)$ for the chance reach probability. Then
\[
c_\sigma=\sum_{\substack{h\in H:\ s_1(h)=\sigma}}
\pi_0(h)y_{s_2(h)}z_{s_3(h)}g_1(h).
\]
For a feasible plan with $x_\sigma=1$, these weights are the reach probabilities of a subset of terminal histories. Their sum is at most one, giving $|c_\sigma|\le M_1$.

\begin{proposition}
\label{pr:lambda}
Let $d_I$ denote the number of player 1 information sets in the
subtree rooted at information set $I$ (including $I$ and descendants through all its actions). Then, for every
player 1 information set $I$,
\[
|\lambda^1_I| \le d_I(M_1 + R_1).
\]

Moreover,
\[
u_1^{\min} \le \lambda^1_{\emptyset} \le u_1^{\max}.
\]
\end{proposition}

\begin{proof}
Let $\sigma$ end with an action at $I$; $\operatorname{child}(\sigma)$ denotes the information sets with predecessor sequence $\sigma$. By the structure of $\mathbf{E}$,
the stationarity condition
$\mathbf{E}^\top\boldsymbol{\lambda}^1-\mathbf{r}^1=\mathbf{c}$ implies
\[
\lambda_I^1
-
\sum_{J\in\operatorname{child}(\sigma)}
\lambda_J^1
=
c_\sigma+r_\sigma^1.
\]

Equivalently,
\[
\lambda_I^1
=
c_\sigma+r_\sigma^1
+
\sum_{J\in\operatorname{child}(\sigma)}
\lambda_J^1.
\]

By Corollary~\ref{co:slack}, $0 \le r_\sigma^1 \le R_1$.
Since $|c_\sigma|\le M_1$, we have $|c_\sigma+r_\sigma^1|\le M_1+R_1$.

We now proceed by backward induction on the information-set tree. For an
information set with no descendants,
\[
|\lambda_I^1|
=
|c_\sigma+r_\sigma^1|
\le
M_1+R_1.
\]

Suppose that for every descendant information set $J$, $|\lambda_J^1|\le d_J(M_1+R_1)$, where $d_J$ is the number of information sets in the subtree rooted at $J$.
Then
\[
|\lambda_I^1|
\le
(M_1+R_1)
+
\sum_{J\in\operatorname{child}(\sigma)}
d_J(M_1+R_1).
\]

Since the child subtrees are disjoint, the quantity
\[
1+\sum_{J\in\operatorname{child}(\sigma)} d_J
\]%
\noindent is at most $d_I$: it counts descendants through the chosen action only. Hence
\[
|\lambda_I^1|
\le d_I(M_1+R_1).
\]

Finally, strong duality gives $\mathbf{e}^\top\boldsymbol{\lambda}^1 = V_1^{\max}$.
Since $\mathbf{e}$ has a 1 in the initial row and 0 elsewhere,
$\lambda_{\emptyset}^1 = V_1^{\max}$.
As $V_1^{\max}$ is an expected utility,
$u_1^{\min} \le V_1^{\max} \le u_1^{\max}$. Therefore,
\[
u_1^{\min}
\le
\lambda_{\emptyset}^1
\le
u_1^{\max}. \qedhere
\]
\end{proof}

\begin{corollary}
\label{co:lambda}
Let $m_1$ be one plus the number of player 1 information sets. Then, for every row $I$ of $\mathbf{E}$, including the normalization row,
\[
|\lambda^1_I|
\le m_1(M_1+R_1).
\]
\end{corollary}

\begin{proof}
For information-set rows, use $d_I\le m_1$ in Proposition~\ref{pr:lambda}. For the normalization row, $|\lambda^1_\emptyset|\le M_1\le m_1(M_1+R_1)$.
\end{proof}

We can bound $\boldsymbol{\lambda}^2$ and $\boldsymbol{\lambda}^3$ analogously.

\section{Three-player Kuhn poker}
\label{se:kuhn}
Three-player Kuhn poker is a simplified form of limit poker that has been used as a testbed game in the AAAI Annual Computer Poker Competition for several years. There is a single round of betting. Each player first antes a single chip and is dealt a card from a four-card deck that contains a Jack, Queen, King, and Ace. The first player has the option to \emph{bet} a fixed amount of one additional chip (by contrast in \emph{no-limit} games players can bet arbitrary amounts of chips) or to \emph{check} (remain in the hand but not bet an additional chip). When facing a bet, a player can \emph{call} (i.e., match the bet) or \emph{fold} (forfeit the hand). No additional bets or raises beyond the original bet are allowed (while they are allowed in other common poker variants such as Texas hold 'em). If all players but one have folded, then the player who has not folded wins the \emph{pot}, which consists of all chips in the middle. If more than one player has not folded by the end there is a \emph{showdown}, at which point the players reveal their private card and the player with the highest card wins the entire pot (which consists of the initial antes plus all additional bets and calls). The Ace is the highest card, followed by King, Queen, and Jack. As one example of a play of the game, suppose the players are dealt Queen, King, Ace respectively, and player 1 checks, player 2 checks, player 3 bets, player 1 folds, and player 2 calls; then player 3 would win a pot of 5, for a profit of 3 (while player 1 loses 1 and player 2 loses 2).

Note that despite the fact that 3-player Kuhn poker is only a synthetic simplified form of poker and is not actually played competitively, it is still far from trivial to analyze, and contains many of the interesting complexities of popular forms of poker such as Texas hold 'em. First, it is a game of imperfect information, as players are dealt a private card that the other agents do not have access to, which makes the game more complex than a game with perfect information that has the same number of nodes. Despite the size, it is not trivial to compute Nash equilibrium analytically, though an infinite family of Nash equilibria has been computed~\cite{Szafron13:Parametrized}. The equilibrium strategies exhibit the phenomena of \emph{bluffing} (i.e., betting with weak hands such as a Jack or Queen), and \emph{slow-playing} (aka \emph{trapping}) (i.e., checking with strong hands such as a King or Ace in order to induce a bet from a weaker hand). The family of equilibria is based on several parameter values, which once selected determine the probabilities for the other portions of the strategies. One can see that randomization and including some probability on trapping and bluffing are essential in order to have a strong and unpredictable strategy. Thus, while this game may appear quite simple at first glance, analysis is still very far from simple, and the game exhibits many of the complexities of far larger games that are played competitively by humans for large amounts of money. 

Prior work constructs equilibria with several actions assigned probability zero; these \emph{dominated} actions can be removed when searching for an equilibrium~\cite{Szafron13:Parametrized}. These actions include calling a bet with Jack, folding to a bet with Ace, calling a bet with Queen after a bet and a call, and checking with Ace after two players check. We can remove these dominated actions from the game, constructing a reduced game that is guaranteed to contain a Nash equilibrium of the full game as well. The full game has 48 total information sets (16 per player) and 601 total nodes, which includes 288 player decision nodes, 312 terminal nodes, and one chance node. The reduced game after removal of dominated actions has 48 total information sets and 415 total nodes (252 player decision nodes, 162 terminal nodes, and one chance node). So the total number of information sets remains the same while the number of player decision nodes is decreased by 12.5\%.

\section{Experiments}
\label{se:exp}
We used Gurobi 13.0.3~\cite{Gurobi26:Gurobi} with Java 26.0.2 on an Intel Core i7-1065G7 processor (1.30 GHz base clock, 4 cores/8 threads), 16 GB RAM, and 64-bit Windows 11. Each equilibrium solve used seed 92222, 8 threads, and a 300-second time limit; all other solver parameters retained their defaults. In particular, feasibility and optimality tolerances were $10^{-6}$ and integer feasibility tolerance was $10^{-5}$. We ran the configurations sequentially. Each table entry is one run, not an average. Solver time excludes model construction, bound preprocessing, and verification. We also report Gurobi's \emph{work units}, an internal measure of computational effort in arbitrary units that is reproducible for fixed models, hardware, and settings. A wall-clock timeout can stop runs at different work counts.

\subsection{Simplified 3-player Kuhn poker}
Earlier work reported a reduced-game solve in 2.47 seconds with maximum best-response gain $\num{1.4e-17}$~\cite{Ganzfried26:Quadratic}. Our Gurobi 13.0.3 run without finite slack or multiplier bounds took 9.741 seconds (4.348 work units), with full-game gain below $\num{1e-12}$. These historical and current timings are not a controlled comparison. The reduced model contains 73 rows, 249 columns, 169 nonzeros, and 198 quadratic constraints. It adds 22 action-fixing equalities to the full model. Presolve produces 1017 rows, 325 columns, 2185 nonzeros, and 231 bilinear constraints; the solve uses 6,492 nodes and 481,773 simplex iterations. Best responses allow all actions in the original game.

\subsection{Full 3-player Kuhn poker}
\nopagebreak
All full-game configurations use identical objectives and constraints, differing only in variable bounds. The input contains 51 rows, 249 columns, 147 nonzeros, and 198 quadratic constraints. The unbounded model presolves to 2112 rows, 637 continuous columns, 4645 nonzeros, and 492 bilinear constraints. Root relaxation uses 132 iterations. It reaches the 300-second limit without a feasible incumbent. The earlier 24-hour timeout~\cite{Ganzfried26:Quadratic} is historical; our table uses the new matched control.

We compare three choices for each variable family: no finite bounds, uniform corollary bounds, and finer bounds based on the propositions. Raw terminal payoffs lie in $[-2,4]$, giving $R_i=6$, $M_i=4$, and $m_i=17$. Corollaries~\ref{co:slack} and~\ref{co:lambda} therefore give slack cap 6 and multiplier magnitude cap 170. With the retained numerical margin, the implemented caps are 6.000001 and 170.000017. Tensor entries already include the deal probability $1/24$; their extrema must not replace the raw terminal extrema in these corollaries.

Proposition~\ref{pr:lambda} uses the actual descendant count $d_I$. The counts are 4 or 1 for player 1, 3 or 1 for player 2, and 1 for player 3. Thus non-normalization multiplier caps are approximately 40, 30, or 10; normalization multipliers lie in $[-2,4]$. We use $d_I(4+6.000001)$ and normalization bounds $[-2-10^{-6},4+10^{-6}]$.

Proposition~\ref{pr:slack} depends on unknown opponent strategies, so we construct conservative static upper bounds for its right-hand side. For player 1, let $l_t=\sum_{j,k}\min(u_1(t,j,k),0)$ and $h_t=\sum_{j,k}\max(u_1(t,j,k),0)$; then $l_t\le c_t\le h_t$. For a nonempty sequence $\sigma$, let $P$ be the pure realization plans of its root component in the sequence-form forest, and $P_\sigma$ those containing $\sigma$. Define
\[
\begin{aligned}
D(p,q)&=\sum_{t:p_t=1,q_t=0}h_t
          -\sum_{t:p_t=0,q_t=1}l_t,\\
B_\sigma&=\max_{p\in P}\min_{q\in P_\sigma}D(p,q).
\end{aligned}
\]
Other components cancel from $V_1^{\max}-U_\sigma$. Since $c^\top(p-q)\le D(p,q)$ pointwise, taking the same maximum and minimum gives $V_1^{\max}-U_\sigma\le B_\sigma$. We impose $r_\sigma\le\min(6,B_\sigma)+10^{-6}$; the empty-sequence cap is $10^{-6}$. These are globally valid bounds derived from Proposition~\ref{pr:slack}, rather than values evaluated at a previously computed equilibrium. The largest component has nine pure plans, so enumeration is small here. The largest unpadded caps by player are 4, 3.25, and 3. Their calculation and the subtree counts take less than 0.1 seconds per run. We checked the component enumeration and forced-response values against 792 independently solved LPs using pure and mixed opponent profiles.

Table~\ref{tab:results} gives the complete comparison. Corollary~\ref{co:lambda} alone is fastest, at 0.642 seconds (0.252 work units). With both corollaries, the solve takes 16.271 seconds (5.299 work units). Replacing only the multiplier bound by Proposition~\ref{pr:lambda} reduces this to 2.621 seconds (1.090 work units), and replacing only the slack bound reduces it to 3.323 seconds (1.210 work units). Using both refinements takes 10.698 seconds (3.588 work units). In contrast, refining slack bounds alone increases work from 0.678 to 4.681. Across all eight bounded runs, the largest slack and multiplier magnitude were approximately $0.166667$ and $0.541667$, respectively. These values lie well below the uniform caps, despite the slowdown from imposing both corollaries. Tighter domains therefore do not consistently improve search performance.

\begin{table}[!ht]
\centering
\setlength{\belowcaptionskip}{6pt}
\renewcommand{\arraystretch}{1.1}
\caption{Full-game solve seconds, with work units in parentheses. $C_j$: Corollary $j$; $P_j$: proposition-based bounds described in the text. TL denotes timeout without an incumbent.}
\label{tab:results}
\begin{tabular}{|c|r|r|r|}
\hline
 & \multicolumn{3}{c|}{Multiplier bounds} \\
Slack bounds & None & $C_2$ & $P_2$ \\
\hline
None & \shortstack{\rule{0pt}{10pt}TL 300\\(90.707)} & \shortstack{\rule{0pt}{10pt}0.642\\(0.252)} & \shortstack{\rule{0pt}{10pt}2.539\\(1.077)} \\
\hline
$C_1$ & \shortstack{\rule{0pt}{10pt}1.849\\(0.678)} & \shortstack{\rule{0pt}{10pt}16.271\\(5.299)} & \shortstack{\rule{0pt}{10pt}2.621\\(1.090)} \\
\hline
$P_1$ & \shortstack{\rule{0pt}{10pt}13.467\\(4.681)} & \shortstack{\rule{0pt}{10pt}3.323\\(1.210)} & \shortstack{\rule{0pt}{10pt}10.698\\(3.588)} \\
\hline
\end{tabular}
\end{table}

For comparison with the previous solver statistics, the slack-corollary-only model presolves to 1920 rows, 733 columns, 4741 nonzeros, and 396 bilinear constraints, including 96 internally introduced binary variables. It uses 293 root-relaxation iterations, 363 nodes, and 21,276 simplex iterations. Its final cut summary contains 173 cuts: 39 implied bound, 3 MIR, 4 flow cover, 101 RLT, and 26 BQP. This summary counts cuts in the final LP relaxation, not every cut generated. The multiplier-corollary-only model has 1923 presolved rows, 733 columns, 4769 nonzeros, 396 bilinear constraints, and 96 binaries; it uses one node and 2,636 simplex iterations. Both corollaries require 7,886 nodes and 389,542 simplex iterations. Repeated corollary runs reproduce the same work and node counts despite varying elapsed times.

All eight bounded profiles have maximum full-game best-response gain below $\num{1e-12}$ and maximum absolute flow, stationarity, or complementarity residual below $\num{1e-11}$. Best-response LPs use feasibility and optimality tolerances $10^{-9}$ and agree with a separate dynamic-programming calculation within $\num{1e-12}$. These checks are numerical, not exact-arithmetic certificates. Logs, bound vectors, models, and source hashes are retained for reproduction. More games and seeds are needed to assess general performance. For historical context, Gambit 16.4's `logit' method~\cite{Savani25b:Gambit,Turocy10:Computing} returned an approximate equilibrium in 2 minutes and 34 seconds in the earlier comparison; a new comparison requires matched accuracy and stopping criteria.

\section{Conclusion}
\label{se:conc}
Computing Nash equilibrium is a fundamental problem in computational game theory, and many important problems are naturally formulated as multiplayer imperfect-information games. Three-player Kuhn poker has been studied as a well-motivated game in this class. In the earlier comparison of algorithms from the Gambit software suite~\cite{Savani25b:Gambit}, only the logit quantal response approach~\cite{Turocy10:Computing} returned an approximate Nash equilibrium in three-player Kuhn poker, requiring about 2.5 minutes. Recently an approach has been proposed that solves a quadratically constrained program based on a nonlinear complementarity problem (NCP) formulation derived from the sequence-form game representation~\cite{Ganzfried26:Quadratic}. The earlier implementation did not solve the full game within 24 hours, while the reduced game is solved quickly without finite slack or multiplier bounds. These results illustrate both the potential power of removal of dominated actions in imperfect-information games~\cite{Ganzfried26:Dominated}, as well as limitations of the NCP approach. As many games may not have a large number of dominated actions, we would still like to develop an improved exact Nash-equilibrium finding approach that can solve larger games without relying on removal of a large number of dominated actions.

The original NCP declares nonnegative slacks without finite upper bounds and unrestricted multipliers. Tighter variable bounds can have a dramatic effect on the spatial branch-and-bound method used by solvers such as Gurobi. Although the initial variable domains do not specify finite bounds, we are able to derive nontrivial finite bounds for both the slack and multiplier variables by applying strong duality to the best response computation problems from which the NCP formulation is derived. In matched Gurobi 13.0.3 runs, all eight bounded configurations return numerically verified equilibria, while the unbounded formulation times out at 300 seconds. The uniform multiplier bound alone performs best. Finer proposition-based bounds improve some combinations but slow others, illustrating that tighter domains need not improve runtime monotonically. Additional games and seeds are needed to assess the generality of these effects.

\section*{Acknowledgments}
ChatGPT and Codex (OpenAI) were used to assist with checking and revising the mathematical arguments, modifying the experimental implementation, and editing the manuscript. The author takes responsibility for the final content.

\clearpage
\bibliographystyle{plain}
\bibliography{C://FromBackup/Research/refs/dairefs}

\begin{thebibliography}{1}

\bibitem{Ganzfried26:Dominated}
Sam Ganzfried.
\newblock Dominated actions in imperfect-information games.
\newblock In {\em IEEE Conference on Decision and Control (CDC)}, 2026.

\bibitem{Ganzfried26:Quadratic}
Sam Ganzfried.
\newblock Quadratic programming approach for {N}ash equilibrium computation in
  multiplayer imperfect-information games.
\newblock {\em Games}, 17(1):9, 2026.

\bibitem{Gurobi26:Gurobi}
{Gurobi Optimization, LLC}.
\newblock {Reference Manual}, 2026.

\bibitem{Koller94:Fast}
Daphne Koller, Nimrod Megiddo, and Bernhard {von Stengel}.
\newblock Fast algorithms for finding randomized strategies in game trees.
\newblock In {\em Proceedings of the {ACM} {S}ymposium on {T}heory of
  {C}omputing}, 1994.

\bibitem{Lemke65:Bimatrix}
C.~E. Lemke.
\newblock Bimatrix equilibrium points and mathematical programming.
\newblock {\em Management Science}, 11(7):681--689, May 1965.

\bibitem{Lemke64:Equilibrium}
C.~E. Lemke and J.~T. Howson.
\newblock Equilibrium points of bimatrix games.
\newblock {\em Journal of the Society for Industrial and Applied Mathematics},
  12(2):413--423, 1964.

\bibitem{Savani25b:Gambit}
Rahul Savani and Theodore~L. Turocy.
\newblock {Gambit: The Package for Computation in Game Theory, {V}ersion
  16.4.0}, 2025.

\bibitem{Szafron13:Parametrized}
Duane Szafron, Richard Gibson, and Nathan Sturtevant.
\newblock A parameterized family of equilibrium profiles for three-player
  {K}uhn poker.
\newblock In {\em {Proceedings of the International Conference on Autonomous
  Agents and Multi-Agent Systems (AAMAS)}}, 2013.

\bibitem{Turocy10:Computing}
Theodore~L. Turocy.
\newblock Computing sequential equilibria using agent quantal response
  equilibria.
\newblock {\em Economic Theory}, 42(1):255--269, 2010.

\end{thebibliography}

\end{document}